\pdfoutput=1

\documentclass[numbers]{sigplanconf}
\usepackage[numbers,sort&compress]{natbib}
\usepackage[english]{babel}
\usepackage[T1]{fontenc}
\usepackage[utf8]{inputenc} 
\usepackage{upgreek}
\usepackage{paralist}
\usepackage{enumitem}
\usepackage{url}
\usepackage{stmaryrd}
\usepackage{graphicx}
\usepackage{infer}

\usepackage{xspace}
\usepackage{amsmath,amsfonts,amsthm,amssymb}
\usepackage{xcolor}
\usepackage[ruled,vlined]{algorithm2e}
\usepackage{nicefrac}
\usepackage{bbm}

\newcommand{\EasyCrypt}{\textsf{EasyCrypt}\xspace}

\newcommand{\Saprhl}{\textsf{apRHL}\xspace}
\newcommand{\Saprhlp}{\textsf{apRHL}$^+$\xspace}

\usepackage{mathpartir}
\usepackage{yfonts}

\usepackage{todonotes}

\usepackage{xargs}
\newcommandx{\note}[2][1=]{\todo[inline,linecolor=red,backgroundcolor=red!25,bordercolor=red,#1]{#2}}

\newtheorem{thm}{Theorem}

\newtheorem{definition}{Definition}
\newtheorem{prop}{Proposition}

\newcommand{\rname}[1]{[{\sc #1}]}

\newcommand{\Dist}{\ensuremath{\mathbf{Distr}}}
\newcommand{\SDist}{\ensuremath{\mathbf{SDistr}}}

\newcommand{\supp}{\mathsf{supp}}

\newcommand{\subst}[2]{\left\{#2/#1\right\}}

\newcommand{\Var}{\mathcal{X}}

\newcommand\q{[\![}
\newcommand\p{]\!]}

\newcommand{\aprhl}[5]{\vdash #1 \sim_{#5} #2 : #3 \Longrightarrow #4}
\newcommand{\AEquiv}[6]{\vdash {#2} \sim_{\!\left\langle#5,#6\right\rangle} {#3} : {#1} \Longrightarrow {#4}}
\newcommand{\sidel}{\langle 1\rangle}
\newcommand{\sider}{\langle 2\rangle}

\newcommand\denot[1]{\q #1 \p}
\newcommand{\dsem}[2]{\denot{#2}_{#1}}

\newcommand{\Skip}{\mathsf{skip}}
\newcommand{\Seq}[2]{{#1};\,{#2}}
\newcommand{\Ass}[2]{#1 \leftarrow #2}
\newcommand{\Rand}[2]{#1 \stackrel{\raisebox{-.25ex}[.25ex]%
 {\tiny $\mathdollar$}}{\raisebox{-.2ex}[.2ex]{$\leftarrow$}} #2}
\newcommand{\Cond}[3]{\mathsf{if}\ #1\ \mathsf{then}\ #2\ \mathsf{else}\ #3}

\newcommand{\WWhile}[2]{\mathsf{while}\ #1\ \mathsf{do}\ #2}

\newcommand{\Lap}{\mathcal{L}}
\newcommand{\OneLap}{\mathcal{L}^{\mathsf{os}}}
\newcommand{\ExpMech}{\mathsf{ExpM}}
\newcommand{\Expr}{\mathcal{E}}

\newcommand{\Cmd}{\mathcal{C}}
\newcommand{\qscore}{\mathsf{qscore}}

\newcommand{\lift}[1]{#1^\sharp}
\newcommand{\alift}[2]{#1^{\sharp #2}}
\usepackage{listings}

\usepackage[T1]{fontenc}
\usepackage{microtype}

\usepackage[scaled]{beramono}
\newcommand\Small{\fontsize{8.2pt}{8.4pt}\selectfont}
\newcommand*\LSTfont{\Small\ttfamily\SetTracking{encoding=*}{-60}\lsstyle}

\def\lstrnd{\stackrel{\raisebox{-.15ex}{\ensuremath{\scriptscriptstyle\$}}}{\raisebox{-.2ex}{\ensuremath{\leftarrow}}}}

\usepackage{xcolor}
\definecolor{DarkGreen}{rgb}{0.1,0.5,0.1}
\definecolor{DarkRed}{rgb}{0.5,0.1,0.1}
\definecolor{DarkBlue}{rgb}{0.1,0.1,0.5}
\usepackage{hyperref}
\hypersetup{
    unicode=false,          
    pdftoolbar=true,        
    pdfmenubar=true,        
    pdffitwindow=false,      
    pdftitle={},    
    pdfauthor={}
    pdfsubject={},   
    pdfnewwindow=true,      
    pdfkeywords={keywords}, 
    colorlinks=true,       
    linkcolor=DarkRed,          
    citecolor=DarkGreen,        
    filecolor=DarkRed,      
    urlcolor=DarkBlue,          
}
\usepackage[capitalise]{cleveref}

\crefname{section}{\S}{\S}
\Crefname{section}{\S}{\S}

\crefname{prop}{proposition}{propositions}
\Crefname{prop}{Proposition}{Propositions}

\crefname{lem}{lemma}{lemmas}
\Crefname{lem}{Lemma}{Lemmas}

\crefname{thm}{theorem}{theorems}
\Crefname{thm}{Theorem}{Theorems}

\crefname{definition}{definition}{definitions}
\Crefname{definition}{Definition}{Definitions}

\newif\iftpdp\tpdpfalse

\begin{document}

\conferenceinfo{LICS '16}{July 5--8, 2016, New York, New York, USA}
\copyrightyear{2016}
\copyrightdata{978-1-4503-4391-6/16/07}
\copyrightdoi{2933575.2934554}
\publicationrights{licensed}     

\title{Proving Differential Privacy via Probabilistic Couplings}
\authorinfo{
  Gilles Barthe$^{\star}$ \and
  Marco Gaboardi$^{\ddagger}$ \and
  Benjamin Gr\'egoire$^{\$}$ \and
  Justin Hsu$^{\#}$ \and
  Pierre-Yves Strub$^{\star}$
}{
$^{\star}$ IMDEA Software \and
$^{\ddagger}$ University at Buffalo, SUNY \and
$^{\$}$ Inria \and
$^{\#}$ University of Pennsylvania
}{}
\maketitle

\begin{abstract}
  Over the last decade, \emph{differential privacy} has achieved
  widespread adoption within the privacy community. Moreover, it has
  attracted significant attention from the verification community,
  resulting in several successful tools for formally proving
  differential privacy.  Although their technical approaches vary
  greatly, all existing tools rely on reasoning principles derived
  from the \emph{composition theorem} of differential privacy. While
  this suffices to verify most common private algorithms, there are
  several important algorithms whose privacy analysis does not rely
  solely on the composition theorem.  Their proofs are significantly
  more complex, and are currently beyond the reach of verification
  tools.

  In this paper, we develop compositional methods for formally verifying
  differential privacy for algorithms whose analysis goes
  beyond the composition theorem.  Our methods are based on
  deep connections between  differential privacy and
  \emph{probabilistic couplings}, an established
  mathematical tool for reasoning about stochastic processes. Even
  when the composition theorem is not helpful, we can often prove
  privacy by a coupling argument.

  We demonstrate our methods on two algorithms: the \emph{Exponential
    mechanism} and the \emph{Above Threshold} algorithm, the critical
  component of the famous \emph{Sparse Vector} algorithm. We verify
  these examples in a relational program logic \Saprhlp, which can
  construct approximate couplings. This logic extends the existing
  \Saprhl logic with more general rules for the Laplace mechanism and
  the one-sided Laplace mechanism, and new structural rules enabling
  pointwise reasoning about privacy; all the rules are inspired by the
  connection with coupling. While our paper is presented from a formal
  verification perspective, we believe that its main insight is of
  independent interest for the differential privacy community.
\end{abstract}


\section{Introduction}
\emph{Differential privacy} is a rigorous definition of statistical
privacy proposed by Dwork, McSherry, Nissim and Smith~\citep{DMNS06},
and considered to be the gold standard for privacy-preserving
computations.  Most differentially private computations are built from
two fundamental tools: private primitives and composition theorems
\iftpdp\else(see \Cref{sec:diffpriv})\fi{}. However, there are several important
examples whose privacy proofs go beyond these tools, for instance:
\begin{itemize}
\item The \emph{Above Threshold} algorithm, which takes a list of
  numerical queries as input and outputs the first query
  whose answer is above a certain threshold. Above Threshold is the
  critical component of the Sparse Vector technique. (See, e.g.,
  \citet{DR14}.)

\item The \emph{Report-noisy-max} algorithm, which takes a list of
  numerical queries as input and privately selects the query with the
  highest answer.  (See, e.g., \citet{DR14}.)

\item The \emph{Exponential mechanism}~\citep{MT07}, which privately
  returns the element of a (possibly non-numeric) range with the
  highest score; this algorithm can be implemented as a variant of the
  Report-noisy-max algorithm with a different noise distribution.
\end{itemize}
Unfortunately, existing pen-and-paper proofs of these algorithms use
ad hoc manipulations of probabilities, and as a consequence are
difficult to understand and error-prone.

This raises a natural question: can we develop \emph{compositional
  proof methods} for verifying differential privacy of these
algorithms, even though their proofs appear non-compositional?
Surprisingly, the answer is yes. Our method builds on two key
insights.
\begin{enumerate}
\item A connection between probabilistic liftings and probabilistic
  couplings~\citep{BartheEGHSS15}.
\item A connection between differential privacy and \emph{approximate
    liftings}~\citep{BartheKOZ13,BartheO13}, a generalization of probabilistic
  liftings used in probabilistic process algebra~\citep{Jonsson:2001}.
\end{enumerate}

\subsection*{Probabilistic liftings and couplings}
\emph{Relation lifting} is a well-studied construction in mathematics and
computer science. Abstractly, relation lifting transforms
relations $R \subseteq A \times B$ into relations $\lift{R} \subseteq
TA \times TB$, where $T$ is a functor over sets~\citep{Barr70}. Relation
lifting satisfies a type of composition, so it is a natural foundation for
compositional proof methods.

Relation lifting has historically been an important tool for analyzing
of probabilistic systems. For example, \emph{probabilistic lifting}
specializes the notion of relation lifting for the probability
monad, and appears in standard definitions of probabilistic
bisimulation. Over the last 25 years, researchers have developed a
wide variety of tools for reasoning about probabilistic liftings,
explored applications in numerous areas including security and
biology, and uncovered deep connections with the Kantorovich metric and
the theory of optimal transport (for a survey, see \citet{DengD11}).

While research has traditionally considers probabilistic liftings for partial
equivalence relations, recent works investigate liftings for more general
relations. Applications include formalizing reduction-based cryptographic
proofs~\citep{BartheGZ09} and modeling stochastic dominance and convergence of
probabilistic processes~\citep{BartheEGHSS15}. Seeking to explain the power of
liftings, \citet{BartheEGHSS15} establish a tight connection between
probabilistic liftings and \emph{probabilistic couplings}, a basic tool in
probability theory~\citep{Lindvall02,Thorisson00}. Roughly, a probabilistic
coupling places two distributions in the same probabilistic space by exhibiting
a suitable \emph{witness distribution} over pairs.  Not only does this
observation open new uses for probabilistic liftings, it offers an opportunity
to revisit existing applications from a fresh perspective.

\subsection*{Differential privacy via approximate probabilistic liftings}
Relational program logics~\citep{BartheKOZ13,BartheO13} and relational
refinement type systems~\citep{BGGHRS15} are currently the most flexible techniques
for reasoning formally about differentially private computations. Their
expressive power stems from \emph{approximate probabilistic liftings}, a
generalization of probabilistic liftings involving a metric on
distributions. In particular, differential privacy is a consequence of a
particular form of approximate lifting.

These approaches have successfully verified differential privacy for
many algorithms. However, they are unsuccessful when privacy does not
follow from standard tools and composition properties. In fact, the
present authors had long believed that the verification of such
examples was beyond the capabilities of lifting-based methods.

\subsection*{Contributions}
In this paper, we propose the first formal analysis of differentially
private algorithms whose proof does not exclusively rely on the
basic tools of differential privacy. We make three broad
contributions.

\paragraph*{New proof principles for approximate liftings}  
We take inspiration from the connection between liftings and coupling
to develop new proof principles for approximate liftings.

First, we introduce a principle for decomposing proofs of differential
privacy \emph{pointwise}, supporting a common pattern of proving
privacy separately for each possible output value. This principle is
used in pen-and-paper proofs, but is new to formal approaches.

Second, we provide new proof principles for the Laplace
mechanism. Informally speaking, existing proof principles capture the
intuition that different inputs can be made to look equal by the
Laplace mechanism in exchange for paying some privacy cost.  Our
first new proof principle for the Laplace mechanism is dual, and
captures the idea that equal inputs can be made to look arbitrarily
\emph{different} by the Laplace mechanism, provided that one pays
sufficient privacy. Our second new proof principle for the Laplace
mechanism states that if we add the same noise in two runs of the
Laplace mechanism, the distance between the two values is preserved
and there is no privacy cost.  As far as we know, these proof
principles are new to the differential privacy literature. They are
the key ingredients to proving examples such as Sparse Vector using
compositional proof methods.

We also propose approximate probabilistic liftings for the one-sided Laplace
mechanism, which can be used to implement the Exponential mechanism. The
one-sided Laplace mechanism nicely illustrates the benefits of our approach:
although it is not differentially private, its properties can be formally captured
by approximate probabilistic liftings. These properties can be combined
to show privacy for a larger program.

\paragraph*{An extended probabilistic relational program logic}
To demonstrate our techniques, we work with the relational program logic \\
\Saprhl~\citep{BartheKOZ13}. Conceived as a probabilistic
variant of Benton's relational Hoare logic~\citep{Benton04}, \Saprhl has been
used to verify differential privacy for examples using the standard composition
theorems. Most importantly, the semantics of \Saprhl uses approximate
liftings.  We introduce new proof rules representing our new proof principles,
and call the resulting logic \Saprhlp.

\paragraph*{New privacy proofs}
While the extensions amount to just a handful of rules, they significantly
increase the power of \Saprhl: We provide the first formal verification of two
algorithms whose privacy proofs use tools beyond the composition theorems.
\begin{itemize}
  \item The \emph{Exponential mechanism}. The standard private
    algorithm when the output is non-numeric, this construction is
    typically taken as a primitive in systems verifying privacy.
    In contrast, we prove its privacy within our logic.

  \item The \emph{Sparse Vector} algorithm. Perhaps the most famous
    example not covered by existing techniques, the proof of this
    mechanism is quite involved; some of its variants are not
    provably private. We also prove the privacy of its core subroutine in our
    logic.
\end{itemize}
The proofs are based on coupling ideas, which avoid reasoning about
probabilities explicitly.  As a consequence, proofs are clean,
concise, and, we believe, appealing to researchers from both the
differential privacy and the formal verification communities.

We have formalized the proofs of these algorithms in an
experimental branch of the \EasyCrypt proof assistant supporting
approximate probabilistic liftings.

\iftpdp\else
\section{Differential privacy}\label{sec:diffpriv}
In this section, we review the basic tools of differential privacy,
and we present the algorithm Above Threshold, which forms the main
subroutine of the Sparse Vector algorithm.

\subsection{Basics}
The basic definition of differential privacy is due to \citet{DMNS06}.
\begin{definition}[Differential privacy] 
A probabilistic computation $M:A\rightarrow\Dist(B)$ satisfies
$(\epsilon,\delta)$-\emph{differential privacy} w.r.t.\, an adjacency
relation $\Phi \subseteq A \times A$ if for every pair of inputs $a, a'\in A$
such that $a~\Phi~a'$ and every subset of outputs $E \subseteq B$, we have
\[
  \Pr_{y\leftarrow M a}[y \in E]
  \leq \exp(\epsilon) \Pr_{y\leftarrow M a'} {[y \in E]} + \delta .
\]
When $\delta = 0$, we say that $M$ is $\epsilon$-\emph{differentially private}.
\end{definition}
Intuitively, the probabilistic condition ensures that any two inputs
satisfying the adjacency relation $\Phi$ result in similar
distributions over outputs. The relation $\Phi$ models which pairs of
databases should be protected, i.e., what data should be nearly
indistinguishable. While it may not be obvious from the definition,
differential privacy has a number of features that allow simple
construction of private algorithms with straightforward proofs of
privacy. Specifically, the vast majority of differential privacy
proofs use two basic tools: private primitives and composition
theorems.

\paragraph*{Private primitives}
These components form the building blocks of private algorithms.  The most basic
example is the \emph{Laplace mechanism}, which achieves differential privacy for
numerical computations by adding probabilistic noise to the output. We will work
with the discrete version of this mechanism throughout the paper.
\begin{definition}[Laplace mechanism \citep{DMNS06}]
  Let $\epsilon>0$. The \emph{(discrete) Laplace mechanism}
  $\Lap_\epsilon:\mathbb{Z}\rightarrow \SDist(\mathbb{Z})$ is defined by
  $\Lap_{\epsilon}(t) = t + \nu$,
  where $\nu \in \mathbb{Z}$ is drawn from the Laplace distribution
  $\mathrm{Laplace}(1/\epsilon)$, i.e.\, with probabilities proportional to
  \[
    \Pr[ \nu ] \propto \exp{(-\epsilon\cdot |\nu|)}.
  \]
\end{definition}
The level of privacy depends on the sensitivity of the query.
\begin{definition}[Sensitivity]
  Let $k\in\mathbb{N}$. A function $F:A\rightarrow\mathbb{Z}$ is
  \emph{$k$-sensitive with respect to $\Phi\subseteq A\times A$} if
  $|F(a_1) - F(a_2)| \leq k$ for every $a_1,a_2\in A$ such that
  $a_1\ \Phi\ a_2$.
\end{definition}
The following theorem shows that $k$-sensitive functions can be made
differentially private through the Laplace mechanism~\citep{DMNS06}.
\begin{thm} \label{thm:lap:priv}
  Assume that $F:A\rightarrow\mathbb{Z}$ is $k$-sensitive with respect
  to $\Phi$. Let $M:A\rightarrow\Dist(\mathbb{Z})$ be the probabilistic
  function that maps $a$ to $\Lap_{\epsilon}(F(a))$. Then $M$ is
  $k\cdot \epsilon$-differentially private with respect to $\Phi$.
\end{thm}

Another private primitive is the Exponential mechanism, which is the tool of
choice when the desired output is non-numeric. While this mechanism is often
taken as a primitive construct, we will see in \Cref{sec:exp-mech} how to verify its privacy.

\paragraph*{Composition theorems} These tools prove the
privacy of a combination of private components, significantly
simplifying the privacy analysis. The most commonly instance,
by far, is the powerful \emph{sequential composition theorem}.
\begin{thm}[Sequential composition~\citep{dwork2006our}]
  Let $M : D \to \Dist(R)$ be an $(\epsilon,\delta)$-private computation, and
  let $M': D \to R \to \Dist(R')$ be an $(\epsilon',\delta')$-private
  computation in the first argument for any fixed value of the second argument.
  Then, the function
  \[
    d \mapsto bind\ M(d)\ M'(d)
  \]
  is $(\epsilon + \epsilon', \delta + \delta')$-private.
\end{thm}
One specific form of composition is post-processing. Informally, the
post-processing theorem states that the output of a differentially
private computation can be transformed while remaining private, so
long as the transformation does not depend on the private data
directly; such a transformation can be thought of as $(0,0)$-differentially
private.

\begin{figure}
\[
  \begin{array}{l}
    \Ass{i}{1}; \Ass{r}{|Q|+1};\\
    \Rand{T}{\Lap_{\epsilon/2}(t)}; \\
    \WWhile{i<|Q|}{} \\
    \quad \Rand{S}{\Lap_{\epsilon/4}(\mathsf{evalQ}(Q[i],d))};  \\
    \quad \mathsf{if}~ (T\leq S~ \land r = |Q|+ 1)~\mathsf{then}~\Ass{r}{i}; \\
    \quad \Ass{i}{i + 1}; \\
    \mathsf{return}~r
  \end{array}
\]
\caption{The Above Threshold algorithm}\label{fig:abovet}
\end{figure}

\subsection{Above Threshold}
\label{ssec:sv-intro}
While most private algorithms can be analyzed using composition
theorems and proofs of private primitives, some algorithms require
more intricate proofs. To give an example, we consider the Above
Threshold algorithm, which is the core of the Sparse
Vector technique.\iftpdp\else\footnote{%
  As this algorithm was not formally proposed
  in a canonical work, there exist different variants of the
  algorithm. Some variants take as input a stream rather than a list
  of queries, and/or output the result of a noisy query, rather than
  its index; see the final remark in~\Cref{app:sparse} for further
  discussion.}\fi{}
The Sparse Vector algorithm takes as input a database
$d$, a list of numerical queries $Q$, a threshold $t$, and a natural
number $k$, and privately selects the first $k$ queries from $Q$ whose
output on $d$ are approximately above the threshold. The Above
Threshold algorithm corresponds to the case $k=1$.

The code of the algorithm is given in \Cref{fig:abovet}. In words,
\textsf{AboveT} computes a noisy version $T$ of the threshold $t$,
computes for every query $q$ in the list $Q$ a noisy version $S$ of
$q(d)$, and returns the index of the first query $q$ such that $T\leq
S$ or a default value if there is no such query. It is easy to see
that $(\epsilon,0)$-differential privacy of \textsf{AboveT} directly
implies $(k\cdot \epsilon, 0)$-differential privacy of Sparse Vector,
since we can simply run \textsf{AboveT} $k$ times in sequence and
apply the sequential composition theorem.

If we try applying the sequential composition theorem (with the privacy of the
Laplace mechanism) to \textsf{AboveT} we can show $(|Q|\cdot \epsilon,
0)$-differential privacy when all queries in $Q$ are $1$-sensitive, where $|Q|$
denotes the length of the list $Q$. However, a sophisticated analysis gives a
more precise privacy guarantee.
\begin{thm}[see, e.g., \citet{DR14}]
Assuming all queries in $Q$ are $1$-sensitive, \textsf{AboveT} is
$(\epsilon,0)$-differentially private.
\end{thm}
In other words, \textsf{AboveT} is provably $\epsilon$-differentially
private, \emph{independent of the number of queries}. This is a
remarkable feature of the Above Threshold algorithm.
\fi

\section{Generalized probabilistic liftings}
To verify advanced algorithms like \textsf{AboveT}, we will leverage the power
of \emph{approximate probabilistic liftings}. In a nutshell, our proofs will
replace the sequential composition theorem of differential privacy---which we've
seen is not enough to verify our target examples---with the more general
composition principle of liftings.  This section reviews existing notions of
(approximate) probabilistic liftings and introduces proof principles for
establishing their existence. Most of these proof principles are new, including
those for equality (\Cref{prop:pw:eq}), differential privacy
(\Cref{prop:pw:dp}), the Laplace mechanism
(\Cref{prop:lap:gen,prop:lap:nogrow}), and the one-sided Laplace mechanism
\iftpdp\else(\Cref{prop:olap:gen,prop:olap:nogrow})\fi.

\iftpdp\else
To avoid measure-theoretic issues, we base our technical development
on sub-distributions over discrete sets (\emph{discrete
  sub-distributions}). For simplicity, we will work with distributions
over the integers when considering distributions over numeric values.

We start by reviewing the standard definition of sub-distributions.
Let $B$ be a countable set. A function $\mu:B\rightarrow \mathbb{R}^{\geq 0}$ is
\begin{itemize}
\item a \emph{sub-distribution} over $B$ if
  $\sum_{b\in\supp(\mu)} \mu(b)\leq 1$; and

\item a \emph{distribution} over $B$ if
  $\sum_{b\in\supp(\mu)} \mu(b)= 1$.
\end{itemize}
As usual, the \emph{support} $\supp(\mu)$ is the subset of $B$ with
non-zero weight under $\mu$.  Let $\Dist(B)$ and $\SDist(B)$ denote
the sets of discrete sub-distributions and distributions respectively
over $B$. Equality of distributions is defined as pointwise equality
of functions.

Probabilistic liftings and couplings are defined in terms of a distribution over
products, and its
\emph{marginal distributions}. Formally, the first and second marginals
of a sub-distribution $\mu\in\Dist(B_1\times B_2)$ are simply the
projections: the sub-distributions $\pi_1(\mu)\in\Dist(B_1)$ and
$\pi_2(\mu)\in\Dist(B_2)$ given by
\[
\pi_1(\mu)(b_1)=\sum_{b_2\in B_2} \mu(b_1,b_2) \qquad
\pi_2(\mu)(b_2)=\sum_{b_1\in B_1} \mu(b_1,b_2) .
\]
\fi

\subsection{Probabilistic couplings and liftings}
Probabilistic couplings and liftings are standard tools 
in probability theory, and semantics and verification,
respectively. We
present their definitions to highlight their similarities
before discussing some useful consequences.
\begin{definition}[Coupling]
There is a \emph{coupling} between two sub-distributions
$\mu_1\in\Dist(B_1)$ and $\mu_2\in\Dist(B_2)$ if there exists a sub-distribution
(called the \emph{witness}) $\mu\in\Dist (B_1\times B_2)$ s.t.\,
$\pi_1(\mu)=\mu_1$ and $\pi_2(\mu)=\mu_2$. 
\end{definition} 
Probabilistic liftings are a special class of couplings.
\begin{definition}[Lifting]
Two sub-distributions $\mu_1\in\Dist(B_1)$ and $\mu_2\in\Dist(B_2)$
are related by the \emph{(probabilistic) lifting} of $\Psi\subseteq B_1\times
B_2$, written
$\mu_1 \lift{\Psi} \mu_2$, if there exists a coupling $\mu\in
\Dist(B_1\times B_2)$ of $\mu_1$ and $\mu_2$ such that $\supp(\mu)
\subseteq \Psi$.
\end{definition}
Probabilistic liftings have many useful consequences. For example,
$\mu_1~\lift{=}~\mu_2$ holds exactly when the
sub-distributions $\mu_1$ and $\mu_2$ are equal. Less trivially,
liftings can bound the probability of one event
by the probability of another event. This observation is useful for formalizing
reduction-based cryptographic proofs.
\begin{prop}[\citet{BartheGZ09}]\label{prop:imp} Let $E_1\subseteq B_1$,
$E_2\subseteq B_2$, $\mu_1\in\Dist(B_1)$ and $\mu_2\in\Dist(B_2)$.
Define
\[
\Psi = \{ (x_1,x_2) \in B_1\times B_2 \mid x_1\in E_1 \Rightarrow x_2\in E_2 \}
.
\]
If $\mu_1 \lift{\Psi} \mu_2$, then
\[
  \Pr_{x_1\leftarrow\mu_1} [x_1\in E_1] \leq \Pr_{x_2\leftarrow\mu_2} [x_2\in
  E_2] .
\]
\end{prop}
One key observation for our approach is that this result can also be used to
prove equality between distributions in a pointwise style.
\begin{prop}[Equality by pointwise lifting]\label{prop:pw:eq}\mbox{}
\begin{itemize}
\item Let $\mu_1,\mu_2\in\SDist(B)$. For every $b\in B$,
define
$$\Psi_b =\{ (x_1,x_2)\in B\times B\mid ~ x_1= b \Rightarrow x_2 = b\} .$$
If $\mu_1 ~\lift{\Psi_b}~ \mu_2$ for all $b\in B$, then $\mu_1 = \mu_2$.

\item Let $\mu_1,\mu_2\in\Dist(B)$. For every $b\in B$, define
\[
\Psi_b =\{ (x_1,x_2)\in B\times B\mid ~ x_1= b \Leftrightarrow x_2 = b\} .
\]
If $\mu_1 ~\lift{\Psi_b} ~ \mu_2$ for all $b\in B$, then $\mu_1 = \mu_2$.
\end{itemize}
\end{prop}
\iftpdp\else
\begin{proof}
We prove the first item; the second item follows similarly.

First, a simple observation: two distributions $\mu_1$
and $\mu_2$ are equal iff $\mu_1(b)\leq \mu_2(b)$ for every $b\in
B$. Indeed, suppose that $\mu_1(\bar{b})\neq \mu_2(\bar{b})$ for some
$\bar{b}\in B$.  Then, $\mu_1(\bar{b})< \mu_2(\bar{b})$, so
$$
\sum_{b\in B} \mu_1(b) <\sum_{b\in B} \mu_2(b) ,
$$
contradicting the fact that $\mu_1$ and $\mu_2$ are distributions:
$$
\sum_{b\in B} \mu_1(b)=\sum_{b\in B} \mu_2(b)=1 .
$$
Thus, in order to show $\mu_1=\mu_2$, it is sufficient to prove
$\Pr_{x\leftarrow \mu_1} {[x=b]} \leq \Pr_{x\leftarrow \mu_2}
{[x=b]}$ for every $b\in B$. These inequalities follow from
\Cref{prop:imp}.
\end{proof}
\fi

\subsection{Approximate liftings}
It has previously been shown that differential privacy follows from an
approximate version of liftings~\citep{BartheKOZ13}. Our
presentation follows subsequent refinements by~\citet{BartheO13}. We start by
defining a notion of distance between sub-distributions. 
\begin{definition}[\citet{BartheKOZ13}]
Let $\epsilon\geq 0$. The $\epsilon$-\emph{DP divergence}
$\Delta_{\epsilon}(\mu_1,\mu_2)$ between two
sub-distributions $\mu_1\in\Dist(B)$ and $\mu_2\in\Dist(B)$ is defined as
$$\sup_{E\subseteq B}
\left(\Pr_{x\leftarrow \mu_1}[x\in E] - \exp(\epsilon) 
\Pr_{x\leftarrow \mu_2} {[x\in E]}\right)$$
\end{definition}

The following proposition relates $\epsilon$-DP divergence with $(\epsilon,
\delta)$-differential privacy.
\begin{prop}[\citet{BartheKOZ13}]\label{prop:dp:div}\mbox{}
A probabilistic computation $M:A\rightarrow\Dist(B)$ is
$(\epsilon, \delta)$-differentially private w.r.t.\, an adjacency relation
$\Phi$ iff 
$$\Delta_{\epsilon}(M(a),M(a'))\leq\delta$$ 
for every two adjacent inputs $a$ and $a'$ (i.e.\, such that $a~\Phi~a'$). 
\end{prop}
We can use DP-divergence to define an approximate version of probabilistic
lifting, called $(\epsilon,\delta)$-\emph{lifting}. We adopt the definition
by~\citet{BartheO13}, which extends to a general class of distances called
$f$-divergences.
\begin{definition}[$(\epsilon,\delta)$-lifting] \label{def:approx-lift}
Two sub-distributions $\mu_1\in\Dist(B_1)$ and $\mu_2\in\Dist(B_2)$
are related by the $(\epsilon,\delta)$-\emph{lifting} of
$\Psi\subseteq B_1\times B_2$, written $\mu_1
\alift{\Psi}{(\epsilon,\delta)} \mu_2$, if there exist two witness
sub-distributions $\mu_L\in\Dist(B_1\times B_2)$ and
$\mu_R\in\Dist(B_1\times B_2)$ such that
\begin{enumerate}
\item $\pi_1(\mu_L)=\mu_1$ and $\pi_2(\mu_R)=\mu_2$;
\item $\supp(\mu_L)\subseteq \Psi$ and  $\supp(\mu_R)\subseteq \Psi$; and
\item $\Delta_\epsilon(\mu_L,\mu_R)\leq\delta$.
\end{enumerate}
\end{definition}
It is relatively easy to see that two sub-distributions $\mu_1$ and
$\mu_2$ are related by $\alift{=}{(\epsilon,\delta)}$ iff
$\Delta_\epsilon(\mu_1,\mu_2)\leq\delta$. Therefore, a probabilistic
computation $M:A\rightarrow\Dist(B)$ is $(\epsilon,
\delta)$-differentially private w.r.t.\, an adjacency relation $\Phi$
iff
\[
  M(a)~\alift{=}{(\epsilon, \delta)}~M(a')
\]
for every two adjacent inputs $a$ and $a'$ (i.e.\, such that
$a~\Phi~a'$).  This fact forms the basis of previous lifting-based
approaches for differential privacy
\citep{BartheKOZ13,BartheO13,BGGHRS15,BGGHKS14}.

A useful preliminary fact is that approximate liftings generalize probabilistic
liftings (which we will sometimes call \emph{exact} liftings).

\begin{prop} \label{prop:alift:lift}
Suppose we are given distributions $\mu_1\in\SDist(B_1)$ and
$\mu_2\in\SDist(B_2)$ and a relation $\Psi\subseteq B_1\times B_2$.
Then, $\mu_1 \lift{\Psi} \mu_2$ if and only if $\mu_1 \alift{\Psi}{(0,0)}
\mu_2$.
\end{prop}
\iftpdp\else
\begin{proof}
  The forward direction is easy: simply define $\mu_L = \mu_R$ to be
  the witness of the exact lift. The reverse direction follows from
  the observations that the witnesses $\mu_L$ and $\mu_R$ are
  necessarily distributions, and that $\Delta_0$ is the total
  variation distance (a.k.a.\, statistical distance) on distributions,
  in particular $\Delta_0(\mu_L,\mu_R)=0$ iff $\mu_L =\mu_R$. To see
  this last point, $\Delta_0(\mu_L,\mu_R)=0$ entails
  \[
    \mu_L(b_1, b_2) \leq \mu_R(b_1, b_2)
  \]
  for every $(b_1, b_2) \in B_1 \times B_2$. So $\mu_L = \mu_R$ by
  \Cref{prop:pw:eq}.
\end{proof}
\fi

The previous results for exact liftings generalize smoothly to
approximate liftings. First, we can generalize \Cref{prop:imp}.
\begin{prop}[\citet{BartheO13}]\label{prop:imp:epsdel} Let 
$E_1\subseteq B_1$, $E_2\subseteq B_2$, $\mu_1\in\Dist(B_1)$ and
  $\mu_2\in\Dist(B_2)$.  Let
$$\Psi = \{ (x_1,x_2) \in B_1\times B_2 \mid x_1\in E_1 \Rightarrow
  x_2\in E_2 \} .$$
If $\mu_1 \alift{\Psi}{(\epsilon,\delta)} \mu_2$, then
$$\Pr_{x_1\leftarrow\mu_1} [x_1\in E_1] \leq \exp(\epsilon)
\Pr_{x_2\leftarrow\mu_2} [x_2\in E_2] +\delta .$$
\end{prop}
We can use this proposition to generalize~\Cref{prop:pw:eq}, which provides a
way to prove that two distributions $\mu_1$ and $\mu_2$ are equal---equivalently,
$\mu_1 \lift{=} \mu_2$. Generalizing this
lifting from exact to approximate yields the following pointwise
characterization of differential privacy, a staple technique of pen-and-paper proofs.
\begin{prop}[Differential privacy from pointwise lifting]\label{prop:pw:dp}
A probabilistic computation $M:A\rightarrow\Dist(B)$ is $(\epsilon,
\delta)$-differentially private w.r.t.\, an adjacency relation $\Phi$
iff there exists $(\delta_b)_{b \in B}\in\mathbb{R}^{\geq 0}$ such
that $\sum_{b\in B} \delta_b \leq\delta$, and
$M(a)~\alift{\Psi_b}{(\epsilon,\delta_b)} ~M(a')$ for every $b\in B$ and
every two adjacent inputs $a$ and $a'$, where
$$\Psi_b =\{ (x_1,x_2)\in B\times B \mid x_1= b\Rightarrow x_2 = b\} .$$
\end{prop}
\iftpdp\else
\begin{proof}
First note that $\Delta_{\epsilon}(\mu_1,\mu_2)\leq\delta$ iff there
exists $(\delta_b)_{b\in B}\in\mathbb{R}^{\geq 0}$ s.t. $\mu_1 (b)\leq
\exp(\epsilon) \mu_2(b) +\delta_b$ for every $b\in B$, and $\sum_{b\in
  B}\delta_b\leq\delta$.
So, it is sufficient to show that for every $b\in B$ and every
two adjacent inputs $a$ and $a'$, we have
$$\Pr_{x\leftarrow M(a)} [x=b] \leq \exp(\epsilon)
\Pr_{x\leftarrow M(a')} [x=b] +\delta_b$$
with $\sum_{b\in B} \delta_b \leq\delta$. This follows from
\Cref{prop:imp:epsdel}.
\end{proof}
\fi

\subsection{Probabilistic liftings for the Laplace mechanism}
So far, we have seen general properties about approximate liftings and
differential privacy. Now, we turn to more specific liftings relevant to
typical distributions in differential privacy.
In terms of approximate liftings, we can state the privacy of the Laplace
mechanism \iftpdp\else(\Cref{thm:lap:priv})\fi{} in the following form.
\begin{prop}\label{prop:lap}
Let $v_1,v_2\in\mathbb{Z}$ and $k\in\mathbb{N}$ s.t. $|v_1 - v_2 |
\leq k$. Then $\Lap_\epsilon(v_1) ~\alift{=}{(k\cdot \epsilon,0)}~
\Lap_\epsilon (v_2)$.
\end{prop} 

\Cref{prop:lap} is sufficiently general to capture most examples
from the literature, but not for the examples of this paper; informally,
applying \Cref{prop:lap} only allows us to prove privacy using the standard
composition theorems. To see how we might generalize the principle, note that
privacy from pointwise liftings (\Cref{prop:pw:dp}) involves liftings
of an \emph{asymmetric} relation, rather than equality. This suggests that it
could be profitable to consider asymmetric liftings. Indeed, we propose the
following generalization of \Cref{prop:lap}.
\begin{prop}\label{prop:lap:gen}
Let $v_1,v_2,k\in\mathbb{Z}$.  Then
\[
  \Lap_\epsilon (v_1)
  ~\alift{\Psi}{(|k + v_1 - v_2|\cdot \epsilon,0)}~
  \Lap_\epsilon (v_2) ,
\]
where
\[
  \Psi = \{ (x_1,x_2)\in\mathbb{Z}\times\mathbb{Z} \mid x_1 + k = x_2 \} .
\]
\end{prop} 
\iftpdp\else
\begin{proof}
  It suffices to prove $\mu_1 ~\alift{\Psi}{(|k + v_1 - v_2|\cdot
    \epsilon,0)}~ \mu_2$, where $\mu_1$ is the distribution of $v_1 + \eta_1 +
  k$ and $\mu_2$ is the distribution of $v_2 + \eta_2$, with $\eta_1, \eta_2$
  draws from the discrete Laplace distribution $\text{Laplace}(1/\epsilon)$. By
  the definition of the Laplace mechanism, $\mu_1 = \Lap_\epsilon(v_1 + k)$ and
  $\mu_2 = \Lap_\epsilon(v_2)$. Now, we can conclude by \Cref{prop:lap}.
\end{proof}
\fi
\Cref{prop:lap:gen} has several useful consequences. For instance, when $|v_1 -
v_2| \leq k$ we have $\Lap_\epsilon(v_1) ~\alift{\Psi}{(2k\cdot \epsilon,0)}~
\Lap_\epsilon (v_2)$ with
\begin{equation} \label{eq:lapapart}
\Psi = \{ (x_1,x_2)\in\mathbb{Z}\times\mathbb{Z} \mid
x_1+k=x_2 \} ,
\end{equation}
following from \Cref{prop:lap:gen} and the triangle inequality
\[
  |v_1 - v_2| \leq k \Rightarrow |k + (v_1 - v_2)| \leq k + k = 2k .
\]
Informally, this instance of \Cref{prop:lap:gen} shows that by ``paying''
privacy cost $\epsilon$, we can ensure that the samples are a certain distance
apart. This stands in contrast to \Cref{prop:lap}, which ensures that the
samples are equal.

Another useful consequence is that adding identical noise to both
$v_1$ and $v_2$ incurs no privacy cost, and we can assume the
difference between the samples is the difference between $v_1$ and
$v_2$.
\begin{prop}\label{prop:lap:nogrow}
Let $v_1,v_2\in\mathbb{Z}$. Then $\Lap_\epsilon(v_1) ~\alift{\Psi}{(0,0)}~
\Lap_\epsilon (v_2)$, where
\[
  \Psi = \{ (x_1,x_2)\in\mathbb{Z}\times\mathbb{Z} \mid x_1-x_2 = v_1 - v_2 \} .
\]
\end{prop}
\iftpdp\else
\begin{proof}
  Immediate by \Cref{prop:lap:gen} with $k = v_2 - v_1$.
\end{proof}
\fi

\iftpdp\else
\subsection{Probabilistic liftings for one-sided Laplace mechanism}

While the Laplace mechanism is already sufficient to implement a wide variety of
private algorithms, a few algorithms use other
distributions. In particular, the Exponential mechanism can be
implemented in terms of the \emph{one-sided Laplace} mechanism. This algorithm
is the same as the Laplace mechanism except noise is drawn from the \emph{one-sided
Laplace distribution} (also called the \emph{exponential distribution}), which
outputs non-negative integers.

\begin{definition}[One-sided Laplace mechanism]
Let $\epsilon>0$. The \emph{discrete one-sided Laplace mechanism}
$\OneLap_\epsilon:\mathbb{Z}\rightarrow \SDist(\mathbb{Z})$ is defined by
$$ 
\OneLap_{\epsilon}(t) = t + \nu,
$$
where $\nu$ non-negative integer drawn from the Laplace distribution
$\mathrm{Laplace}^+(1/\epsilon)$, i.e.\, with probabilities proportional to
\[
  \Pr[ \nu ] \propto \exp{(-\epsilon \cdot \nu)}.
\]
\end{definition}

While this mechanism is not $\epsilon$-differentially private for any
$\epsilon$, we can
still consider probabilistic liftings for the samples. We have the following two
results, analogous to \Cref{prop:lap:gen,prop:lap:nogrow}.

\begin{prop}\label{prop:olap:gen}
Let $v_1,v_2,k\in\mathbb{Z}$ such that $k \geq v_2 - v_1$. Then
\[
  \OneLap_\epsilon (v_1)
  ~\alift{\Psi}{((k + v_1 - v_2) \cdot \epsilon,0)}~
  \OneLap_\epsilon (v_2) ,
\]
where
\[
  \Psi = \{ (x_1,x_2) \in \mathbb{Z} \times \mathbb{Z} \mid x_1 + k = x_2 \} .
\]
\end{prop} 
\begin{proof}
  It suffices to consider the case where $v_1 = v_2 = 0$:
  $\OneLap_\epsilon (v)$ is the same distribution as sampling from
  $\OneLap_\epsilon(0)$ and adding $v$, so the desired conclusion follows from
  \[
    \OneLap_\epsilon (0)
    ~\alift{\Psi'}{((k + v_1 - v_2) \cdot \epsilon,0)}~
    \OneLap_\epsilon (0) ,
  \]
  where
  \begin{align*}
    \Psi' &= \{ (x_1,x_2) \in \mathbb{Z} \times \mathbb{Z} \mid (x_1 + v_1) + k
    = (x_2 + v_2) \} \\
    &= \{ (x_1,x_2) \in \mathbb{Z} \times \mathbb{Z} \mid x_1 + (k + v_1 - v_2)
    = x_2 \} ,
  \end{align*}
  which follows from the $v_1 = v_2 = 0$ case since $k + v_1 - v_2 \geq 0$ by
  assumption.

  So, we assume $v_1 = v_2 = 0$ and $k \geq 0$.  We will directly define the two
  witnesses of the approximate lifting. Let
  \[
    G(v) = \Pr_{x\leftarrow \OneLap_\epsilon(0)}[x = v] .
  \]
  Define the left witness $\mu_L$ on its support by
  \[
    \mu_L(i, i + k) = G(i)
  \]
  for $i \geq 0$, and the right witness $\mu_R$ on its support by
  \[
    \mu_R(j - k, j) = G(j)
  \]
  for $j \geq 0$. Evidently the marginals are correct---$\pi_1(\mu_L) =
  \pi_2(\mu_R) = \OneLap_\epsilon(0)$---so it remains to check that
  $\Delta_{k\epsilon}(\mu_L, \mu_R) \leq 0$:
  \[
    \max_{E\subseteq \mathbb{Z} \times \mathbb{Z}}
    \left(\Pr_{(x, y)\leftarrow \mu_L}[(x, y)\in E] - e^{k\epsilon}
    \Pr_{(x, y)\leftarrow \mu_R} {[(x, y)\in E]}\right) \leq 0.
  \]
  It suffices to prove this pointwise over the union of the supports of
  $\mu_L$ and $\mu_R$: for each $l \geq -k$, we need
  \[
    \mu_L(l, l + k) - e^{k\epsilon} \mu_R(l, l + k) \leq 0.
  \]
  This is evident for $l < 0$, when the first term is zero and the second
  term is non-negative. For $l \geq 0$ we need to show
  \[
    G(l) - e^{k\epsilon} G(l + k) \leq 0 ,
  \]
  which follows by direct calculation (or, the privacy of the standard
  Laplace distribution).
\end{proof}

\begin{prop}\label{prop:olap:nogrow}
Let $v_1,v_2\in\mathbb{Z}$. Then $\OneLap_\epsilon(v_1) ~\alift{\Psi}{(0,0)}~
\OneLap_\epsilon (v_2)$, where
\[
  \Psi = \{ (x_1,x_2) \in \mathbb{Z} \times \mathbb{Z} \mid x_1 - x_2 = v_1 - v_2 \} .
\]
\end{prop}
\begin{proof}
  It suffices to prove
  \[
    \OneLap_\epsilon(v_1) ~\alift{\Psi'}{(0,0)}~ \OneLap_\epsilon (v_2) ,
  \]
  where
  \[
    \Psi' = \{ (x_1,x_2) \in \mathbb{Z} \times \mathbb{Z} \mid x_1 - v_1 = x_2 - v_2 \} .
  \]
  This is equivalent to
  \[
    \OneLap_\epsilon(v_1 - v_1) ~\alift{=}{(0,0)}~ \OneLap_\epsilon (v_2 - v_2) ,
  \]
  which is obvious by \Cref{prop:alift:lift} since both sides are the same
  distribution.
\end{proof}
\fi

\iftpdp\else
\begin{figure*}[t]
$$
\begin{array}{c@{}}
 { }
 {
   \AEquiv{\Psi\subst{x_1\sidel,x_2\sider}{e_1\sidel,e_2\sider}}
     {\Ass{x_1}{e_1}}{\Ass{x_2}{e_2}}{\Psi}{0}{0}
   }
[\textsc{Assn}]
\\[4ex]

\infrule
{\AEquiv{\Phi \land b_1\sidel}{c_1}{c_2}{\Psi}{\epsilon}{\delta} 
\quad 
 \AEquiv{\Phi \land \lnot b_1\sidel}{d_1}{d_2}{\Psi}{\epsilon}{\delta} 
}
  {\AEquiv{\Phi\land b_1\sidel = b_2\sider}{\Cond{b_1}{c_1}{d_1}}{
\Cond{b_2}{c_2}{d_2}}{\Psi}{\epsilon}{\delta}}
[\textsc{Cond}]\\[4ex]
\infrule
{
 \AEquiv{\Theta \land b_1\sidel \land b_2\sider \land k = e\sidel \land e\sidel \leq n}
         {c_1}{c_2}{\Theta \land b_1\sidel = b_2\sider \land k < e\sidel}
         {\epsilon_k}{\delta_k} \hspace{1cm}
 \Theta \land e\sidel \leq 0 \Rightarrow \neg b_1\sidel 
}{
 \AEquiv{\Theta \land b_1\sidel = b_2\sider \land e\sidel \leq n}
         {\WWhile{b_1}{c_1}}{\WWhile{b_2}{c_2}}
         {\Theta \land \neg b_1\sidel \land\neg b_2 \sider}
         {\sum_{k=1}^{n}\epsilon_k}{\sum_{k=1}^{n}\delta_k}
}
[\textsc{While}] \\[4ex]

\infrule
{\AEquiv{\Phi}{c_1}{c_2}{\Psi'}{\epsilon}{\delta} 
  \quad \AEquiv{\Psi'}{c_1'}{c_2'}{\Psi}{\epsilon'}{\delta'} 
}
{\AEquiv{\Phi}{c_1;c_1'}{c_2;c_2'}{\Psi}
  {\epsilon + \epsilon'}{\delta + \delta'}}
[\textsc{Seq}]\\[4ex]

\infrule {
  \AEquiv{\Phi'}{c_1}{c_2}{\Psi'}{\epsilon'}{\delta'} \qquad
  \Phi \Rightarrow \Phi'  \qquad \Psi' \Rightarrow \Psi \qquad
  \epsilon' \leq \epsilon \qquad \delta' \leq \delta
  }
 {
   \AEquiv{\Phi}{c_1}{c_2}{\Psi}{\epsilon}{\delta}
 }
[\textsc{Conseq}] 
\end{array}
$$
\caption{Proof rules from \Saprhl}\label{fig:aprhl}
\end{figure*}
\fi

\section{Formalization in a program logic}
\label{sec:apRHL}
In this section we present a new program logic called \Saprhlp for
reasoning about differential privacy of programs written in a core
programming language with samplings from the Laplace mechanism and the
one-sided Laplace Mechanism. Our program logic \Saprhlp extends
\Saprhl, a relational Hoare logic that has been used to verify many
examples of differentially private algorithms~\citep{BartheKOZ13}.
\iftpdp\else
The main result of this section is a proof of soundness of the logic
(\Cref{thm:sound}).
\fi

\iftpdp
\else
\paragraph*{Programs} We consider a simple imperative language with 
random sampling. The set of commands is defined inductively:
\begin{displaymath}
\begin{array}{r@{\ \ }l@{\quad}l}
\Cmd ::= & \Skip                   & \mbox{noop} \\
     \mid& \Seq{\Cmd}{\Cmd}        & \mbox{sequencing}\\
     \mid& \Ass{\Var}{\Expr}       & \mbox{deterministic assignment}\\
     \mid& \Rand{\Var}{\Lap_\epsilon(\Expr)}     
                                   & \mbox{Laplace mechanism}\\
     \mid& \Rand{\Var}{\OneLap_\epsilon(\Expr)}     
                                   & \mbox{one-sided Laplace mechanism}\\
     \mid& \Cond{\Expr}{\Cmd}{\Cmd} & \mbox{conditional}\\
     \mid& \WWhile{\Expr}{\Cmd}      & \mbox{while loop}
\end{array}
\end{displaymath}
where $\Var$ is a set of \emph{variables} and $\Expr$ is a set of
\emph{expressions}. Variables and expressions are typed, and range
over boolean, integers, databases, queries, and lists.
\fi

The semantics of programs is standard~\citep{Kozen79,BartheKOZ13}. We
first define the set $\mathsf{Mem}$ of memories to contain all
well-typed functions from variables to values.
\iftpdp
\else
Expressions and
distribution expressions map memories to values and distributions over
values, respectively: an expression $e$ of type $T$ is interpreted as
a function $\denot{e}:\mathsf{Mem}\rightarrow T$, whereas a
distribution expression $g$ is interpreted as a function
$\denot{g}:\mathsf{Mem} \rightarrow\SDist(\mathbb{Z})$. Finally,
commands are interpreted as functions from memories to
sub-distributions over memories, i.e.\, the interpretation of $c$ is a
function $\denot{c}: \mathsf{Mem} \rightarrow \Dist(\mathsf{Mem})$. We
refer to~\citet{Kozen79,BartheKOZ13} for an account of the semantics.
\fi

\paragraph*{Assertions and judgments}
Assertions in the logic are first-order formulae over generalized
expressions. The latter are expressions built from tagged variables
$x\sidel$ and $x\sider$, where the tag is used to determine whether
the interpretation of the variable is taken in the first memory or in
the second memory. For instance, $x\sidel=x\sider +1$ is the assertion
which states that the interpretation of the variable $x$ in the first
memory is equal to the interpretation of the variable $x$ in the
second memory plus 1. More formally, assertions are interpreted as
predicates over pairs of memories. We let $\denot{\Phi}$ denote the
set of memories $(m_1,m_2)$ that satisfy $\Phi$. The interpretation is
standard (besides the use of tagged variables) and is omitted. By
abuse of notation, we write $e\sidel$ or $e\sider$, where $e$ is a
program expression, to denote the generalized expression built
according to $e$, but in which all variables are tagged with a
$\sidel$ or $\sider$, respectively.

Judgments in both \Saprhl and \Saprhlp are of the form
$$\AEquiv{\Phi}{c_1}{c_2}{\Psi}{\epsilon}{\delta}$$
where $c_1$ and $c_2$ are statements, the precondition $\Phi$ and
postcondition $\Psi$ are relational assertions, and $\epsilon$ and
$\delta$ are non-negative reals.\iftpdp\else\footnote{%
  The original \Saprhl rules are based on a multiplicative privacy budget.
  We adapt the rules to an additive privacy parameter for consistency with the
  rest of the article and the broader privacy literature.}\fi{}
Informally, a judgment of the above form is valid if the two
distributions produced by the executions of $c_1$ and $c_2$ on any two initial
memories satisfying the precondition $\Phi$ are related by the
$(\epsilon,\delta)$-lifting of the postcondition $\Psi$. Formally,
the judgment
$$\AEquiv{\Phi}{c_1}{c_2}{\Psi}{\epsilon}{\delta}$$
is \emph{valid} iff for every two memories $m_1$ and $m_2$, such that
$m_1~\denot{\Phi}~m_2$, we have
$$
(\dsem{m_1}{c_1})
~\alift{\denot{\Psi}}{(\epsilon,\delta)} ~
(\dsem{m_2}{c_2}) .
$$

\begin{figure*}[t]
$$\begin{array}{c}
\infrule{\forall i. \AEquiv{\Phi}{c_1}{c_2}{x\sidel = i \Rightarrow
    x\sider = i}{\epsilon}{
\delta_i}\hspace{2cm} \sum_{i\in I} \delta_i\leq \delta}{
\AEquiv{\Phi}{c_1}{c_2}{x\sidel=x\sider}{\epsilon}{\delta}}[\textsc{Forall-Eq}]
\\[4ex]
\infrule {
  }
  {\AEquiv{|k + e_1\sidel - e_2\sider| \leq k'}
   {\Rand{y_1}{\Lap_\epsilon(e_1)}}{\Rand{y_2}{\Lap_\epsilon(e_2)}}{y_1\sidel+k=y_2\sider}
   {k' \cdot \epsilon}{0}}
[\textsc{LapGen}]
\\[4ex]
\infrule {y_1 \notin FV(e_1) \qquad y_2 \notin FV(e_2)}
  {\AEquiv{\top}
   {\Rand{y_1}{\Lap_\epsilon(e_1)}}{\Rand{y_2}{\Lap_\epsilon(e_2)}}{
   y_1\sidel - y_2\sider = e_1\sidel - e_2\sider}
   {0}{0}}
[\textsc{LapNull}]
\iftpdp\else
\\[4ex]
\infrule {
  }
  {\AEquiv{0 \leq k + e_1\sidel -e_2\sider \leq k'}
   {\Rand{y_1}{\OneLap_\epsilon(e_1)}}{\Rand{y_2}{\OneLap_\epsilon(e_2)}}{y_1\sidel
     + k = y_2\sider}
   {k' \cdot \epsilon}{0}}
[\textsc{OneLapGen}]
\\[4ex]
\infrule {y_1 \notin FV(e_1) \qquad y_2 \notin FV(e_2) }
  {\AEquiv{\top}
   {\Rand{y_1}{\OneLap_\epsilon(e_1)}}{\Rand{y_2}{\OneLap_\epsilon(e_2)}}{
   y_1\sidel - y_2\sider = e_1\sidel - e_2\sider}
   {0}{0}}
[\textsc{OneLapNull}]

\\[4ex]

\infrule{\AEquiv{\Phi\land b_1\sidel}{c_1}{c}{\Psi}{\epsilon}{\delta}\quad
\AEquiv{\Phi\land \neg b_1\sidel}{d_1}{c}{\Psi}{\epsilon}{\delta}}{
\AEquiv{\Phi}{\Cond{b_1}{c_1}{d_1}}{c}{\Psi}{\epsilon}{\delta}
}[\textsc{Cond-L}]

\\[4ex]
\infrule{\AEquiv{\Phi\land b_2\sider}{c}{c_2}{\Psi}{\epsilon}{\delta}\quad
\AEquiv{\Phi\land \neg b_2\sider}{c}{d_2}{\Psi}{\epsilon}{\delta}}{
\AEquiv{\Phi}{c}{\Cond{b_2}{c_2}{d_2}}{\Psi}{\epsilon}{\delta}
}[\textsc{Cond-R}]
\fi
\end{array}
$$

\iftpdp
\else
\caption{Proof rules from \Saprhlp}
\fi
\label{fig:aprhlp}
\end{figure*}

\paragraph*{Proof system}
\iftpdp
\else
\Cref{fig:aprhl} presents the main rules from \Saprhl excluding
the sampling rule, which we generalize
in \Saprhlp. We briefly comment on some of
these rules.

The rule \rname{Seq} for sequential composition generalizes the sequential
composition theorem of differential privacy, which intuitively
corresponds to the case where the postcondition of the composed
commands is equality. This generalization allows \Saprhl to prove
differential privacy using the coupling composition principle when 
the standard composition theorem is insufficient.

The rule \rname{While} for while loops can be seen as a generalization of a $k$-fold
composition theorem for differential privacy. Again, it allows to
consider arbitrary postconditions, whereas the composition theorem
would correspond to the case where the postcondition of the loop is
equality (in conjunction with negation of the guards). We often use
two simpler instances of the rule. The first one corresponds to the
case where the values of $\epsilon_k$ and $\delta_k$ are independent
of $k$, i.e.\, $\epsilon_k=\epsilon$ and $\delta_k=\delta$, yielding a
bound of $\langle n\cdot \epsilon, n\cdot \delta\rangle$. The second
one corresponds to the case where a single iteration carries a privacy
cost, as shown in the rule \rname{WhileExt} in \Cref{fig:while}. This weaker rule
is in fact sufficient for proving privacy of several of our examples,
including the Above Threshold algorithm (but not the Sparse Vector algorithm,
which also uses the aforementioned instance of the while rule), the Exponential
mechanism, and Report-noisy-max.
\fi

\Cref{fig:aprhlp} collects the new rules in \Saprhlp, which are all derived from
the new proof principles we saw in the previous section. The first rule
\rname{Forall-Eq} allows proving differential privacy via pointwise privacy;
this rule reflects \Cref{prop:pw:dp}.

The next pair of rules, \rname{LapGen} and \rname{LapNull}, reflect the liftings
of the distributions of the Laplace mechanism presented in
\Cref{prop:lap:gen,prop:lap:nogrow} respectively.  Note that we need a
side-condition on the free variables in \rname{LapNull}---otherwise, the sample
may change $e_1$ and $e_2$.%
\iftpdp\else
The following pair of rules, \rname{OneLapGen} and
\rname{OneLapNull}, give similar liftings for the one-sided Laplace mechanism
following \Cref{prop:olap:gen,prop:olap:nogrow} respectively.

Finally, the last pair of rules allows reasoning about a conditional
while treating the other command abstractly. These so-called
\emph{one-sided} rules were already present in the logic
\textsf{pRHL}, a predecessor of \Saprhl based on exact
liftings~\citep{BartheGZ09}, but they were never needed in \Saprhl. In
\Saprhlp the one-sided rules are quite useful, in conjunction with our
richer sampling rules, for reasoning about two conditionals that may
take different branches.
\fi

\iftpdp\else
\paragraph*{Soundness}
The soundness of the new rules immediately follows from the results of the
previous section, while soundness for the \Saprhl rules was established
previously~\citep{BartheKOZ13}.

\begin{thm}\label{thm:sound}
All judgments derivable in \Saprhlp are valid.
\end{thm}

\begin{figure*}
$$\infrule
{
\begin{array}{l}
 \AEquiv{\Phi \land i <e\sidel}
         {c_1}{c_2}{\Psi}
         {0}{0}
 \qquad
 \AEquiv{\Phi\land e\sidel=i}
         {c_1}{c_2}{\Psi}
         {\epsilon}{\delta}
 \qquad
 \AEquiv{\Phi \land e\sidel<i}
         {c_1}{c_2}{\Psi}
         {0}{0}
 \\[2ex]
 \Theta \land e\sidel \leq 0 \Rightarrow \neg b_1\sidel 
 \qquad
 \Phi \triangleq \Theta \land b_1\sidel \land b_2\sider \land k = e\sidel
 \qquad
 \Psi \triangleq \Theta \land b_1\sidel = b_2\sider \land e\sidel < k
\end{array}
}{
 \AEquiv{\Theta \land b_1\sidel = b_2\sider}
         {\WWhile{b_1}{c_1}}{\WWhile{b_2}{c_2}}
         {\Theta \land \neg b_1\sidel \land\neg b_2 \sider}
         {\epsilon}{\delta}
}[\textsc{WhileExt}]
$$
(Note that the two premises for $i<e\sidel$ and $i>e\sidel$ can be
combined. However, we often use different reasoning for these
cases, so we prefer to present the rule with 3 premises.)
\caption{Specialized proof rule for while loops}\label{fig:while}
\end{figure*}
\fi

\iftpdp\else
\section{Exponential mechanism}
\label{sec:exp-mech}
In this section, we provide a formal proof of the \emph{Exponential mechanism}
of \citet{MT07}. While there is existing work that proves differential privacy
of this mechanism~\citep{BartheKOZ13}, the proofs operate on the raw
denotational semantics. In contrast, we work entirely within our program logic.

The Exponential mechanism is designed to privately compute the best
response from a set $\mathcal{R}$ of possible response, according to
some integer-valued \emph{quality score} function $\qscore$ that takes
as input an element in $\mathcal{R}$ and a database $d$. Given a
database $d$ and a $k$-sensitive quality score function $\qscore$, the Exponential
mechanism $\ExpMech(d,\qscore)$ outputs an element $r$ of the range
$\mathcal{R}$ with probability proportional to
\[
  \Pr[r] \propto \exp{\left(\frac{\epsilon \cdot \qscore(r,d)}{2k}\right)}.
  \]
The shape of the distribution ensures that the Exponential mechanism
favors elements with higher quality scores. 

The seminal result of \citet{MT07} establishes differential privacy for this
mechanism.
\begin{thm}
Assume that the quality score is $1$-sensitive, i.e.\, for every
output $r$ and  adjacent databases $d, d'$,
$$|\qscore(r, d) - \qscore(r, d')| \leq 1 .$$ 
Then the probabilistic computation that maps $d$ to $\ExpMech(d,\qscore)$ is
$(\epsilon,0)$-differentially private.
\end{thm}
While there does not seem to be much of a program to verify, it is
known that the Exponential mechanism can be implemented more explicitly in terms
of the one-sided Laplace mechanism~\citep{DR14}. Informally, the code loops
through all the possible output values, adding one-sided Laplace noise to the
quality score for the value/database pair. Throughout the computation, the code
tracks the current highest noisy score and the corresponding element. Finally,
it returns the top element. For the sake of simplicity we assume that
$\mathcal{R}=\{1,\ldots, R\}$ for some $R\in\mathbb{N}$; generalizing to an
arbitrary finite set poses little difficulty for the verification.
\Cref{fig:expmech} shows the code of the implementation.
\begin{figure}
$$\begin{array}{l}
    \Ass{r}{1}; \Ass{bq}{0}; \\
    \WWhile{r \leq R}{} \\
    \quad \Rand{cq}{\OneLap_{\epsilon/2}(\qscore(d,r))};  \\
    \quad \mathsf{if}~ (cq > bq \vee r=1) ~\mathsf{then} 
~\Ass{\mathit{max}}{r}; \Ass{bq}{cq}; \\
    \quad \Ass{r}{r+1}; \\
    \mathsf{return}~\mathit{max}
  \end{array}
$$
\caption{Implementation of the Exponential mechanism}\label{fig:expmech}
\end{figure}

\paragraph*{Informal proof}
The privacy proof for the Exponential mechanism cannot follow from the
composition theorems of differential privacy---the one-sided Laplace
noise does not satisfy differential privacy, so there is nothing to
compose. Nonetheless, we can still show $(\epsilon,0)$-differential
privacy using our lifting-based techniques. By~\Cref{prop:pw:dp}, it
suffices to show that for every integer $i$ and quality score
$\qscore$, the output of $\ExpMech$ on two adjacent databases yields
sub-distributions on memories that are related by the
$(\epsilon,0)$-lifting of the interpretation of the assertion
$$
\mathit{max}\sidel=i \Rightarrow \mathit{max}\sider =i.
$$
We outline a coupling argument. First, we consider
iterations of the loop body in which the loop counter $r$ satisfies
$r<i$. In this case, we couple the two samplings using the rule
[\textsc{OneLapNull}], using adjacency of the two databases and
$1$-sensitivity of the quality score function to establish the
$(0,0)$-lifting:
$$
\mathit{max}\sidel < i ~\land ~ \mathit{max}\sider < i 
\land |bq\sidel - bq\sider |
\leq 1 .
$$ 
The interesting case is $r=i$. In this case, we use the rule
     [\textsc{OneLapGen}] to couple the random samplings so that
$$
cq\sidel +1 =cq\sider .
$$ 
This coupling has privacy cost $(\epsilon,0)$ and ensures that the
following $(\epsilon,0)$-lifting holds at the end of the $i$th 
iteration:
$$(\mathit{max}\sidel=\mathit{max}\sider = i \land bq\sidel +1
=bq\sider)\vee \mathit{max}\sidel\neq i$$
Using the rule [\textsc{OneLapNull}] repeatedly, we couple the random
samplings from the remaining iterations to prove that the above
$(\epsilon,0)$-lifting remains valid through subsequent
iterations---note that couplings for iterations beyond $i$
incur no privacy cost. Finally, we apply the rule of consequence to
conclude the desired $(\epsilon,0)$-lifting:
$$
\mathit{max}\sidel=i \Rightarrow \mathit{max}\sider =i
$$

\paragraph*{Formal proof}
We prove the following \Saprhlp judgment, which entails
$(\epsilon,0)$-differential privacy:
$$\AEquiv{\Phi}{\ExpMech}{\ExpMech}{
\mathit{max}\sidel= \mathit{max}\sider}{\epsilon}{0}$$
where $\Phi$ denotes the precondition
$$\begin{array}{ll}
& \mathsf{adj}(d\sidel,d\sider) \\
\land & \qscore\sidel =\qscore\sider \\
\land & \forall r\in\mathcal{R}.~
|\qscore\sidel (d\sidel,r) -\qscore\sidel (d\sider,r)|\leq 1 .
\end{array}
$$
The conjuncts of the precondition are self-explanatory: the first
states that the two databases are adjacent, the second
states that the two score functions are equal, and the last
states that the quality score function is $1$-sensitive.

By the rule [\textsc{Forall-Eq}], it suffices to prove
$$
  \AEquiv{\Phi}{\ExpMech}{\ExpMech}
  { (\mathit{max}\sidel\ = i) \Rightarrow (\mathit{max}\sider\ = i)}{\epsilon}{0} .
$$
for every $i\in\mathbb{Z}$.  The main step is to apply the [\textsc{WhileExt}]
rule with a suitably chosen loop invariant $\Theta$. We set $\Theta$ to be
$$(r\sidel <i \Rightarrow \Theta_{<}) \wedge 
(r\sidel \geq i \Rightarrow \Theta_{\geq}) 
\wedge  r\sidel=r\sider ,
$$
where $\Theta_{<}$ stands for
$$
\mathit{max}\sidel < i ~\land ~ \mathit{max}\sider < i \land |bq\sidel - bq\sider |
\leq 1
$$
and $\Theta_{\geq}$ stands for
$$
(\mathit{max}\sidel=\mathit{max}\sider = i \land bq\sidel +1  =bq\sider)\vee
\mathit{max}\sidel\neq i .
$$
Omitting the assertions required for proving termination and
synchronization of the loop iterations (which follows from the
conjunct $r\sidel=r\sider$), we have to prove three different
judgments:
\begin{itemize}
\item case $r<i$: $\AEquiv{r\sidel < i \land \Theta_{<}}{c}{c}{\Theta_<}{0}{0}$
\item case $r=i$: $\AEquiv{r\sidel = i \land \Theta_{<}}{c}{c}{\Theta_{\geq}}{\epsilon}{0}$
\item case $r>i$: $\AEquiv{r\sidel > i \land \Theta_{\geq}}{c}{c}{\Theta_{\geq}}{0}{0}$
\end{itemize}
where $c$ denotes the loop body of $\ExpMech$:
$$
\begin{array}{l}
\Rand{cq}{\OneLap_{\epsilon/2}(\qscore(d,r))};  \\
\mathsf{if}~ (cq > bq \vee r=1) ~\mathsf{then} ~\Ass{\mathit{max}}{r}; \Ass{bq}{cq}; \\
\Ass{r}{r+1}
\end{array}
$$ 
Corresponding conditional statements may take the different branches, so we
apply one sided-rules [\textsc{Cond-L}] and [\textsc{Cond-R}].

\paragraph{Report-noisy-max}
A closely-related mechanism is
\emph{Report-noisy-max} (see, e.g., \citet{DR14}).
This algorithm has the exact same code except that it samples from
the standard (two-sided) Laplace distribution rather than the one-sided Laplace
distribution. It is straightforward to prove privacy for this
modification with the axiom \rname{LapGen} (resp. \rname{LapNull}) for
the standard Laplace distribution in place of \rname{OneLapGen} (resp.
\rname{OneLapNull}).
\fi

\section{Above Threshold algorithm}
\label{app:sparse}
The \emph{Sparse Vector} algorithm is the canonical example of a
program whose privacy proof goes beyond proofs of privacy primitives
and composition theorem. The core of the algorithm is the Above
Threshold algorithm. In this section, we prove that the latter (as
modeled by the program \textsf{AboveT}) is
$(\epsilon,0)$-differentially private; privacy for the full mechanism
follows by sequential composition.

\paragraph*{Informal proof}
By~\Cref{prop:pw:dp}, it suffices to show that for every integer $i$,
the output of $\textsf{AboveT}$ on two adjacent databases yields two
sub-distributions over $\mathsf{Mem}$ that are related by the
$(\epsilon,0)$-lifting of the interpretation of the assertion
$$
r\sidel=i \Rightarrow r\sider =i.
$$ 
The coupling proof goes as follows. We start by coupling the
samplings of the noisy thresholds so that $T\sidel + 1 = T\sider$; the
cost of this coupling is $(\epsilon/2,0)$. For the first $i-1$
queries, we couple the samplings of the noisy query outputs using the
rule [\textsc{LapNull}]. By $1$-sensitivity of the queries and
adjacency of the two databases, we know
$\mathsf{evalQ}(Q[j],d)\sider - \mathsf{evalQ}(Q[j],d) \sidel \leq 1$,
so
$$
S\sidel < T\sidel \Rightarrow S\sider < T\sider .
$$
Thus, if side $\sidel$ does not change the value of $r$, neither
does side $\sider$. In fact, we have the stronger invariant
$$
r\sidel=|Q|+1 \Rightarrow r\sider=|Q|+1 \wedge 
  (r\sidel=|Q|+1 \vee r\sidel < i) ,
$$ 
where $r = |Q| + 1$ means that the loop has not exceeded the threshold yet.

When we reach the $i$th iteration and $i < |Q| + 1$, we couple the
samplings of $S$ so that $S\sidel + 1 = S\sider$; the cost of this
coupling is $(\epsilon/2,0)$. Because $T\sidel + 1 = T\sider$ and $S\sidel
+ 1 = S\sider$, we enter the conditional in the second execution
as soon as we enter the conditional in the first execution.
For the remaining iterations $r > i$, it is easy to prove
$$
r\sidel=i \Rightarrow r\sider =i.
$$

\paragraph*{Formal proof}
We prove the following \Saprhlp judgment, which entails
$(\epsilon,0)$-differential privacy:
$$
\AEquiv{\Phi}{\textsf{AboveT}}{\textsf{AboveT}}{
r\sidel =r\sider}{\epsilon}{0} ,
$$
where $\Phi$ denotes the precondition
$$\begin{array}{ll}
& \mathsf{adj}(d\sidel,d\sider) \\
\land & t\sidel=t\sider \\
\land & Q\sidel=Q\sider \\
\land & \forall j.~ 
\left| \mathsf{evalQ}(Q\sidel [j],d\sidel)- 
\mathsf{evalQ}(Q\sider [j],d\sider) \right|
\leq 1 .
\end{array}
$$
The conjuncts of the precondition are straightforward: the first
states that the two databases are adjacent, the second and
third state that $Q$ and $t$ coincide in both runs, and
the last states that all queries are $1$-sensitive.
By the rule [\textsc{Forall-Eq}], it suffices to prove
$$
  \AEquiv{\Phi}{\textsf{AboveT}}{\textsf{AboveT}}
  { (r\sidel\ = i) \Rightarrow (r\sider\ = i)}{\epsilon}{0} .
$$
for every $i\in\mathbb{Z}$.

We begin with the three initializations:
\[
\begin{array}{l}
 \Ass{j}{1};\\
 \Ass{r}{|Q|+1};\\
 \Rand{T}{\Lap_{\epsilon}(t)}; \\
\end{array}
\]
This command $c_0$ computes a noisy version of the threshold $t$. We use the rule
\rname{LapGen} with $\epsilon=\epsilon/2$, $k=1$ and $k'=k$, noticing
that $t$ is the same value in both sides. This proves the
judgment
\[
  \aprhl{c_0}{c_0}
  { \Phi}
  { T\sidel + 1 = T\sider }{\epsilon/2} .
\]
Notice that the $\epsilon/2$ we are paying here is \emph{not} for the privacy
of the threshold---which is not private information!---but rather for
ensuring that the noisy thresholds are \emph{one apart} in the two runs.

Next, we consider the main loop $c_1$:
\[
  \begin{array}{l}
\WWhile{j<|Q|}{} \\
    \quad \Rand{S}{\Lap_{\epsilon/4}(\mathsf{evalQ}(Q[j],d))};  \\
    \quad \mathsf{if}~ (T\leq S~ \land r = |Q|+ 1)~\mathsf{then}~\Ass{r}{j}; \\
    \quad \Ass{j}{j + 1}; \\
  \end{array}
\]
and prove the judgment
\[
  \aprhl{c_1}{c_1}
  { \Phi\land T\sidel + 1 = T\sider}
  {(r\sidel\ = i) \Rightarrow (r\sider\ = i)}{\epsilon/2}
\]
with the \rname{WhileExt} rule.
\iftpdp\else
The proof is similar to the one
for the Exponential mechanism, using invariants from the informal proof.
\fi

\iftpdp\else
\paragraph{Other versions of Above Threshold}
As noted in the introduction, different versions of Above Threshold have been
considered in the literature. One variant returns the first noisy value above
threshold; see \Cref{fig:above:value} for the code. While this was
thought to be private, errors in the proof were later
uncovered. Under our coupling proof, the error is obvious:
we need to prove
$v\sidel = v\sider$ for the result to be private, so we need
$\mathsf{evalQ}(Q[i],d\sidel) = \mathsf{evalQ}(Q[i],d\sider)$ 
after the critical iteration $r = i$. But we have
already coupled $\mathsf{evalQ}(Q[i],d\sidel)+1 = \mathsf{evalQ}(Q[i],d\sider)$
during this iteration. \citet{lyu2016understanding} provide further discussion
of this, and other, incorrect implementations of the Sparse Vector technique.

On the other hand, it is possible to prove
$(2\epsilon,0)$-differential privacy for a modified version of the
algorithm, where the returned value uses fresh noise (e.g.\, by
adding after the loop has completed the sampling
$\Rand{v}{\Lap_{\epsilon}(\mathsf{evalQ}(Q[r],d))}$).

Another interesting variant of the algorithm deals with streams of
queries. If the output of the queries is
uniformly bounded below, then the program terminates with probability
$1$ and the proof proceeds as usual.  However, if the answers to the
stream of queries are below the threshold and falling, the
probability of non-termination can be positive. The interaction of
non-termination and differential privacy is unusual;
most works assume that algorithms always terminate.

The Sparse Vector technique has also been studied by the database
community. Recent work by \citet{DBLP:journals/corr/ChenM15e} shows
that many proposed generalizations of the Sparse Vector algorithm are
not differentially private.

\begin{figure}
\[
  \begin{array}{l}
    \Ass{i}{1}; \Ass{v}{0};\Ass{r}{|Q|+1};\\
    \Rand{T}{\Lap_{\epsilon/2}(t)}; \\
    \WWhile{i<|Q|}{} \\
    \quad \Rand{S}{\Lap_{\epsilon/4}(\mathsf{evalQ}(Q[i],d))};  \\
    \quad \mathsf{if}~ (T\leq S~ \land r = |Q|+ 1)~\mathsf{then}~
     \Ass{r}{i};~\Ass{v}{S} \\
    \quad \Ass{i}{i + 1}; \\
    \mathsf{return}~v
  \end{array}
\]
\caption{Buggy Above Threshold algorithm}\label{fig:above:value}
\end{figure}

\section{Related work}
Coupling is an established tool in probability theory, but it seems less
familiar to computer science. It was only quite recently that couplings have
been used in cryptography; according to \citet{HoangR10}, who use couplings to
reason about generalized Feistel networks, \citet{Mironov02} first used this
technique in his analysis of RC4. Similarly, we are not aware of couplings in
differential privacy, though there seems to be an implicit coupling argument
by~\citet{DNRR15}. There are seemingly few applications of coupling in formal
verification, despite considerable research on probabilistic bisimulation (first
introduced by~\citet{LarsenS89}) and probabilistic relational program logics
(first introduced by~\citet{BartheGZ09}). The connection between
liftings and couplings was recently noted by~\citet{BartheEGHSS15}.

There are many language-based techniques for proving differential privacy for
programs, including dynamic checking~\citep{pinq,conf/popl/EbadiSS15}, the
already mentioned relational program logic~\citep{BartheKOZ13,BartheO13} and
relational refinement type systems~\citep{BGGHRS15}, linear (dependent) type
systems~\citep{ReedPierce10,GHHNP13}, product programs~\citep{BGGHKS14}, and
methods based on computing bisimulations families for probabilistic
automata~\citep{Tschantz201161,xu:hal-00879140}. None of these techniques can
deal with the examples in this paper.

\section{Conclusion}
We show new methods for proving differential privacy with approximate couplings.
We take advantage of the full generality of approximate couplings, showing that
the composition principle for couplings generalizes the standard composition
principle for differential privacy. Our principles support concise and
compositional proofs that are arguably more elegant than existing pen-and-paper
proofs. Although our results are presented from the perspective of formal
verification, we believe that our contributions are also relevant to the
differential privacy communities.

In the future, we plan to use our methods also for the verification
adaptive data analysis algorithms used to prevent false discoveries,
such as the one proposed by~\citet{DBLP:conf/stoc/DworkFHPRR15}, and
for the formal verification of mechanism design~\citep{BartheGAHRS15}.
Beyond these examples, the pointwise characterization of equality can
be adapted to stochastic dominance, and provides a useful tool to
further investigate machine-checked verification of coupling
arguments.

It could also be interesting to use the pointwise characterization of
differential privacy to simplify existing formal proofs. For example,
\citet{BartheKOZ13} prove differential privacy of the vertex cover
algorithm~\citep{GLMRT10}. This algorithm does not use standard primitives;
instead, it samples from a custom distribution specific to the graph.  The
existing formal proof uses a custom rule for loops, reasoning by case
analysis on the output of the random samplings. Pointwise differential privacy
could handle this reasoning more elegantly.

\paragraph*{Acknowledgments}
We warmly thank Aaron Roth for challenging us with the problem of
verifying Sparse Vector. We also thank him and Jonathan Ullman for good
discussions about challenges and subtleties
of the proof of Sparse Vector. This work was partially supported by NSF grants
TWC-1513694, CNS-1065060 and CNS-1237235, by EPSRC grant EP/M022358/1 and by a
grant from the Simons Foundation ($\#360368$ to Justin Hsu).
\fi

\bibliographystyle{abbrvnat}
\bibliography{header,main}

\end{document}